\documentclass[runningheads]{llncs}
\usepackage{cite}
\usepackage{booktabs}
\usepackage[plain]{algorithm}
\usepackage{subcaption}
\DeclareCaptionSubType*{algorithm}

\DeclareCaptionLabelFormat{alglabel}{Alg.~#2}
\usepackage{xspace}
\usepackage{amssymb}
\usepackage{amsmath}
\usepackage[inline]{enumitem}
\usepackage[binary-units]{siunitx}
\DeclareSIUnit\linesofcode{loc}
\usepackage{tikz}
\usepackage{pgfplots, pgfplotstable}
\usepackage{verbatim}
\usepackage{microtype}

\usepackage{hyperref}
\hypersetup{%
  colorlinks=true,%
  citecolor=blue,%
  urlcolor=blue,%
  linkcolor=blue,%
}

\usetikzlibrary{ipe}

\usepackage[noend]{algpseudocode}
\algnewcommand\algorithmicinput{\textbf{Input:}}
\algnewcommand\INPUT{\item[\algorithmicinput]}
\algnewcommand\algorithmicoutput{\textbf{Output:}}
\algnewcommand\Output{\item[\algorithmicoutput]}
\algrenewcommand\algorithmicindent{1ex}

\usepackage{listings}
\lstdefinelanguage{evm}{
  morekeywords={STOP, ADD, MUL, SUB, DIV, SDIV, MOD, SMOD, ADDMOD, MULMOD, EXP, SIGNEXTEND,
  LT, GT, SLT, SGT, EQ, ISZERO, AND, OR, XOR, NOT, BYTE, SHL, SHR, SAR, SHA3,
  ADDRESS, BALANCE, ORIGIN, CALLER, CALLVALUE, CALLDATALOAD, CALLDATASIZE,
  CALLDATACOPY, CODESIZE, CODECOPY, GASPRICE, EXTCODESIZE, EXTCODECOPY,
  RETURNDATASIZE, RETURNDATACOPY, EXTCODEHASH,
  BLOCKHASH, COINBASE, TIMESTAMP, NUMBER, DIFFICULTY, GASLIMIT,
  POP, MLOAD, MSTORE, MSTORE8, SLOAD, SSTORE, JUMP, JUMPI, PC, MSIZE, GAS, JUMPDEST,
  PUSH, SWAP, SWAP1,SWAP2,SWAP3,SWAP4,SWAP5,SWAP6,SWAP7,SWAP8,SWAP9,SWAP10,SWAP11,
  SWAP12,SWAP13, SWAP14,SWAP15,SWAP16, DUP, DUP1,DUP2,DUP3,DUP4,DUP5,DUP6,DUP7,
  DUP8,DUP9,DUP10,DUP11,DUP12,DUP13, DUP14,DUP15,DUP16,
  LOG0, LOG1, LOG2, LOG3, LOG4, JUMPTO, JUMPIF, JUMPV, JUMPSUB, JUMPSUBV, BEGINSUB, BEGINDATA,
  RETURNSUB, PUTLOCAL, GETLOCAL, CREATE, CALL, CALLCODE, RETURN, DELEGATECALL, CREATE2,
  STATICCALL, REVERT, INVALID, SELFDESTRUCT, LOG},
  sensitive=true,
  alsoletter={0,1,2,3,4,5,6,7,8,9,a,b,c,d,e,f}
}
\begin{document}
\title{Blockchain Superoptimizer\thanks{This research is supported by the UK
    Research Institute in Verified Trustworthy Software Systems and partially
    supported by funding from Google.}}
%
%
\author{Julian Nagele\inst{1}\orcidID{0000-0002-4727-4637}\and
  Maria A Schett\inst{2}\orcidID{0000-0003-2919-5983}}
\authorrunning{J.~Nagele and M.A.~Schett}
%
\institute{Queen Mary University of London, UK
  \email{mail@jnagele.net}\\ \and
  University College London, UK
  \email{mail@maria-a-schett.net}}
\maketitle              
\begin{abstract}
  In the blockchain-based, distributed computing platform Ethereum,
programs called smart contracts are compiled to bytecode and
executed on the Ethereum Virtual Machine (EVM).
Executing EVM bytecode is subject to monetary fees---a clear
optimization target.
Our aim is to superoptimize EVM bytecode by encoding the operational
semantics of EVM instructions as SMT formulas and leveraging a
constraint solver to automatically find cheaper bytecode.
We implement this approach in our EVM Bytecode SuperOptimizer
\textsf{ebso} and perform two large scale evaluations on real-world data
sets.


\end{abstract}
%




\newcommand*{\mas}[1]{\textcolor{green!75!black}{(MS: {#1})}}
\newcommand*{\jn}[1]{\textcolor{blue!75!black}{(JN: {#1})}}
\newcommand*{\gy}[1]{\textcolor{magenta!75!black}{(GY: {#1})}}

\newcommand*{\sh}[1]{\textcolor{black!40}{{#1}}} 

\newcommand*{\github}{\textsf{github}\xspace}
\newcommand*{\google}{\textsf{Google}\xspace}
\newcommand*{\ebso}{\textsf{ebso}\xspace}
\newcommand*{\ethereum}{\textsf{Ethereum}\xspace}
\newcommand*{\evm}{\textsf{EVM}\xspace}
\newcommand*{\etherscan}{\textsf{EtherScan}\xspace}
\newcommand*{\geth}{\textsf{geth}\xspace}
\newcommand*{\isabelle}{\textsf{Isabelle}\xspace}
\newcommand*{\llvm}{\textsf{LLVM}\xspace}
\newcommand*{\hol}{\textsf{HOL}\xspace}
\newcommand*{\svcomp}{\textsf{SV-COMP}\xspace}
\newcommand*{\solidity}{\textsf{Solidity}\xspace}
\newcommand*{\lem}{\textsf{Lem}\xspace}
\newcommand*{\lllang}{\textsf{LLL}\xspace}
\newcommand*{\viper}{\textsf{Vyper}\xspace}
\newcommand*{\solc}{\textsf{solc}\xspace}
\newcommand*{\ZE}{\textsf{Z3}\xspace}
\newcommand*{\CVCF}{\textsf{CVC4}\xspace}
\newcommand*{\vampire}{\textsf{Vampire}\xspace}
\newcommand*{\ocaml}{\textsf{OCaml}\xspace}
\newcommand*{\ounit}{\textsf{OUnit}\xspace}
\newcommand*{\webassembly}{\textsf{WebAssembly}\xspace}
\newcommand*{\alive}{\textsf{Alive}\xspace}

\newcommand{\smt}{SMT\xspace}

\newcommand*{\bso}{\textsf{bso}\xspace}
\newcommand*{\uso}{\textsf{uso}\xspace}
\newcommand*{\usoQ}{\textsf{uso$_?$}\xspace}

\newcommand*{\oc}[1]{{\tt #1}\xspace}

\newcommand*{\pusharg}[1]{#1}

\newcommand*{\PUSH}[1]{\oc{PUSH}\,\pusharg{#1}}
\newcommand*{\POP}{\oc{POP}}
\newcommand*{\ADD}{\oc{ADD}}
\newcommand*{\SUB}{\oc{SUB}}
\newcommand*{\SDIV}{\oc{SDIV}}
\newcommand*{\AND}{\oc{AND}}
\newcommand*{\XOR}{\oc{XOR}}
\newcommand*{\TIMESTAMP}{\oc{TIMESTAMP}}
\newcommand*{\OR}{\oc{OR}}
\newcommand*{\ISZERO}{\oc{ISZERO}}
\newcommand*{\BLOCKHASH}{\oc{BLOCKHASH}}
\newcommand*{\BALANCE}{\oc{BALANCE}}
\newcommand*{\SWAP}[1]{\oc{SWAP#1}}
\newcommand*{\DUP}[1]{\oc{DUP#1}}
\newcommand*{\JUMP}{\oc{JUMP}}
\newcommand*{\JUMPI}{\oc{JUMPI}}
\newcommand*{\JUMPDEST}{\oc{JUMPDEST}}
\newcommand*{\STOP}{\oc{STOP}}
\newcommand*{\EXP}{\oc{EXP}}
\newcommand*{\LT}{\oc{LT}}
\newcommand*{\SLT}{\oc{SLT}}
\newcommand*{\NOT}{\oc{NOT}}
\newcommand*{\ADDRESS}{\oc{ADDRESS}}
\newcommand*{\MLOAD}{\oc{MLOAD}}
\newcommand*{\SLOAD}{\oc{SLOAD}}
\newcommand*{\MSTORE}{\oc{MSTORE}}
\newcommand*{\SSTORE}{\oc{SSTORE}}
\newcommand*{\SHAE}{\oc{SHA3}}
\newcommand*{\LOG}{\oc{LOG}}
\newcommand*{\NUMBER}{\oc{NUMBER}}

\newcommand*{\occ}[1]{{\tt #1}}
\newcommand*{\bsh}[1]{{\tt #1}}

\newcommand*{\da}{$\delta\alpha$}

\newcommand*{\ite}{\textit{ite}}

\newcommand*{\halt}{\mathsf{hlt}}
\newcommand*{\stack}{\mathsf{st}}
\newcommand*{\gas}{\mathsf{g}}
\newcommand*{\stackctr}{\mathsf{c}}
\newcommand*{\inpt}{\vec{x}}
\newcommand*{\Vars}{\mathcal{V}}
\newcommand*{\pres}{\mathsf{pres}}
\newcommand*{\UI}{\mathsf{UI}}
\newcommand*{\ui}{\mathsf{ui}}
\newcommand*{\tmpl}{\mathsf{a}}
\newcommand*{\CI}{\mathsf{CI}}
\newcommand*{\instr}{\mathsf{instr}}
\newcommand*{\storage}{\mathsf{str}}
\newcommand*{\TV}{\nu}

\newcommand*{\ggOpt}{\num{19}}
\newcommand*{\ggMaxTarget}{\lstinline!DUP1 PUSH 0 DUP1 DUP4 MLOAD DUP2 PUSH 0 DUP5!}
\newcommand*{\ggMaxSource}{\lstinline!DUP1 MLOAD DUP2 SWAP1 PUSH 0 SWAP1 DUP2 SWAP1 DUP2 DUP1 DUP1! }
\newcommand*{\ggOptOpt}{\num{10}}
\newcommand*{\ggTotal}{\num{66}}
\newcommand*{\ggRelative}{\SI{22.62}{\percent}}
\newcommand*{\ggBlockcount}{\num{2743}}
\newcommand*{\ggMaxGas}{\SI{9}{\gas}}

\newcommand{\ggData}{{\sf GG}\xspace}
\newcommand{\ggRaw}{{\sf GG_\text{raw}}\xspace}
\newcommand{\ggUso}{{\sf GG_\text{opt}}\xspace}

\newcommand*{\bcUsoOptOpt}{\num{393}}
\newcommand*{\bcUsoMaxTarget}{\lstinline!!}
\newcommand*{\bcBsoMaxGas}{\SI{5220}{\gas}}
\newcommand*{\bcUsoMaxSource}{\lstinline!PUSH 0 PUSH 4 SLOAD SUB PUSH 4 DUP2 SWAP1 SSTORE POP! }
\newcommand*{\bcBsoOnly}{\num{21}}
\newcommand*{\bcBsoTotal}{\num{6903}}
\newcommand*{\bcBsoOpt}{\num{184}}
\newcommand*{\bcUsoTotal}{\num{17871}}
\newcommand*{\bcUsoMaxGas}{\SI{5220}{\gas}}
\newcommand*{\bcUsoOpt}{\num{943}}
\newcommand*{\bcBsoMaxSource}{\lstinline!PUSH 0 PUSH 4 SLOAD SUB PUSH 4 DUP2 SWAP1 SSTORE POP!}
\newcommand*{\bcBsoMaxTarget}{\lstinline!!}
\newcommand*{\bcBlockcount}{\num{61217}}
\newcommand*{\bcUsoRelative}{\SI{29.63}{\percent}}
\newcommand*{\bcBsoRelative}{\SI{46.1}{\percent}}

\newcommand*{\bcData}{{\sf EthBC}\xspace}
\newcommand*{\bcRaw}{{\sf EthBC_{\text{raw}}}\xspace}
\newcommand*{\bcUso}{{\sf EthBC_{\text{\uso}}}\xspace}
\newcommand*{\bcUsoO}{{\sf EthBC_{\text{\uso$\top$}}}\xspace}
\newcommand*{\bcBso}{{\sf EthBC_{\text{\bso}}}\xspace}

\newcommand{\cf}{\emph{cf.}\xspace}
\newcommand{\ie}{\emph{i.e.}\xspace}
\newcommand{\Ie}{\emph{I.e.}\xspace}
\newcommand{\eg}{\emph{e.g.}\xspace}
\newcommand{\Eg}{\emph{E.g.}\xspace}
\newcommand{\vs}{\emph{vs.}\xspace}


\section{Introduction} \label{sec:introduction}

\ethereum is a blockchain-based, distributed computing platform
featuring a quasi-Turing complete programming
language. In \ethereum, programs are called smart
contracts, compiled to bytecode and executed on the Ethereum Virtual
Machine (\evm). In order to avoid network spam and to ensure
termination, execution is subject to monetary
fees. These fees are specified in units of \emph{gas}, \ie, any
instruction executed on the \evm has a cost in terms of gas, possibly
depending on its input and the execution state.

\begin{figure}[t]
  \centering
  \input{ebso-overview}
  \vspace{-1cm}
  \caption{Overview over \ebso.}
  \label{fig:ebso-overview}
\end{figure}

\begin{example} \label{ex:intro} Consider the expression
  $\pusharg{3} + (\pusharg{0} - x)$, which corresponds to the program
  \lstinline!PUSH 0 SUB PUSH 3 ADD!. The \evm is a stack-based
  machine, so this program takes an argument $x$ from the stack to
  compute the expression above.  However, clearly one can save the
  \ADD instruction and instead compute $\pusharg{3} - x$, \ie,
  optimize the program to \lstinline!PUSH 3 SUB!.
  The first program costs \SI{12}{\gas} to execute on the
  \evm, while the second costs only \SI{6}{\gas}.
\end{example}

We build a tool that automatically finds this optimization and similar
others that are missed by state-of-the-art smart contract compilers:
the \textbf{\textsf{E}}\textsf{VM} \textbf{b}ytecode
\textbf{s}uper\textbf{o}ptimizer \ebso.
The use of \ebso for Example~\ref{ex:intro} is sketched in
Figure~\ref{fig:ebso-overview}.
To find these optimizations, \ebso implements
\emph{superoptimization}. Superoptimization is often considered too
slow to use during software development except for special
circumstances. We argue that compiling smart contracts is such a
circumstance. Since bytecode, once it has been deployed to the
blockchain, cannot change again, spending extra time
optimizing a program that may be called many times, might well be worth
it. Especially, since it is very clear what ``worth it'' means: the
clear cost model of gas makes it easy to define optimality.%
\footnote{Of course setting the gas price of individual instructions,
  such that it accurately reflects the computational cost is hard, and
  has been a problem in the past see \eg\
  \href{https://news.ycombinator.com/item?id=12557372}{news.ycombinator.com/item?id=12557372}.}

Our main contributions are:
\begin{enumerate*}[label=\emph{(\roman*)}]
\item an \smt encoding of a subset of \evm bytecode semantics
  (Section~\ref{sec:encoding}),
\item an implementation of two flavors of superoptimization: \emph{basic}, where
  the constraint solver is used to check equivalence of enumerated candidate
  instruction sequences, and \emph{unbounded}, where also the enumeration itself
  is shifted to the constraint solver (Section~\ref{sec:implementation}), and
\item two large scale evaluations (Section~\ref{sec:evaluation}). First, we run
  \ebso on a collection of smart contracts from a programming competition aimed
  at producing the cheapest \evm bytecode for given programming challenges. Even
  in this already highly optimized data set \ebso still finds
  \num{19}~optimizations.  In the second evaluation we compare the performance
  of basic and unbounded superoptimization on the \num{2500} most called smart
  contracts from the \ethereum blockchain and find that, in our setting,
  unbounded superoptimization outperforms basic superoptimization.
\end{enumerate*}


\section{Ethereum and the EVM} \label{sec:preliminaries}

Smart contracts in \ethereum are usually written in a specialized
high-level language such as \solidity or \viper and then compiled
into \emph{bytecode}, which is executed on the \evm.
The \evm is a virtual machine formally defined in the \ethereum yellow
paper~\cite{2018_wood}. It is based on a \emph{stack}, which holds
\emph{words}, \ie, bit vectors, of size \num{256}.\footnote{This word
  size was chosen to facilitate the cryptographic computations such as
  hashing that are often performed in the \evm.}
The maximal \emph{stack size} is set to~$2^{10}$. Pushing words onto a full
stack leads to a \emph{stack overflow}, while removing words from the empty
stack leads to a \emph{stack underflow}. Both lead the \evm to enter an
\emph{exceptional halting} state.
The \evm also features a volatile \emph{memory}, a word-addressed byte array,
and a persistent key-value \emph{storage}, a word-addressed word array, whose
contents are stored on the \ethereum blockchain.
The bytecode directly corresponds to more human-friendly
\emph{instructions}.
For example, the \evm bytecode \oc{6029600101} encodes the following sequence of
instructions: \lstinline!PUSH 41 PUSH 1 ADD!.
Instructions can be classified into different categories, such as
\emph{arithmetic} operations,\eg\ \ADD and \SUB for addition and subtraction,
\emph{comparisons}, \eg\ \SLT\ for signed less-than, and \emph{bitwise
  operations}, like \AND\ and \NOT.
The instruction \lstinline!PUSH! pushes a word onto the stack, while \POP\
removes the top word.\footnote{We gloss over the 32~different \lstinline!PUSH!
  instructions depending on the size of the word to be pushed.}
Words on the stack can be duplicated using \DUP{$i$} and swapped using
\SWAP{$i$} for $1 \leqslant i \leqslant 16$, where $i$ refers to the
$i$th word below the top.
Some instructions are specific to the blockchain domain, like
\BLOCKHASH, which returns the hash of a recently mined block, or
\ADDRESS, which returns the address of the currently executing
account.
Instructions for control flow include \eg\ \JUMP, \JUMPDEST, and
\STOP.

We write $\delta(\iota)$ for the number of words that instruction~$\iota$ takes
from the stack, and $\alpha(\iota)$ for the number of words $\iota$ adds onto the
stack.  A \emph{program} $p$ is a finite sequence of instructions.  We define
the \emph{size} $|p|$ of a program as the number of its instructions.
To execute a program on the \ethereum blockchain, the caller has to pay
\emph{gas}. The amount to be paid depends on both the instructions of the
program and the input: every instruction comes with a \emph{gas cost}. For
example, \lstinline!PUSH! and \ADD\ currently cost \SI{3}{\gas}, and therefore executing
the program above costs \SI{9}{\gas}. Most instructions have a fixed cost, but
some take the current state of the execution into account. A prominent example of
this behavior is storage.  Writing to a zero-valued key conceptually allocates
new storage and thus is more expensive than writing to a key that is already in
use, \ie, holds a non-zero value. The gas prices of all instructions are specified
in the yellow paper~\cite{2018_wood}.




\section{Superoptimization}\label{sec:superopt}

Given a \emph{source} program $p$ superoptimization tries to generate a
\emph{target} program~$p'$ such that
\begin{enumerate*}[label=\emph{(\roman*})]
\item $p'$ is equivalent to $p$, and
\item the cost of $p'$ is minimal with respect to a given cost function $C$.
\end{enumerate*}
This problem arises in several contexts with different source and target
languages. In our case, \ie, for a binary recompiler, both source and target
are \evm bytecode.

A standard approach to superoptimization and synthesis~\cite{1987_massalin,
  2011_gulwani_et_al, 2013_schkufza_et_al, 2015_srinivasan_et_al} is to search
through the space of \emph{candidate instruction} sequences of increasing cost
and use a constraint solver to check whether a candidate correctly implements
the source program.  The solver of choice is usually a Satisfiability Modulo
Theories (\smt) solver, which operates on first-order formulas in combination
with background theories, such as the theory of bit vectors or arrays.
Modern \smt solvers are highly optimized and implement techniques to handle
arbitrary first-order formulas, such as E-matching. With
increasing cost of the candidate sequence, the search space dramatically
increases. To deal with this explosion one idea is to hand some of the search to
the solver, by using templates~\cite{2011_gulwani_et_al, 2015_srinivasan_et_al}.
Templates leave holes in the target program, \eg for immediate arguments of
instructions, that the solver must then fill.  A candidate program is correct if
the encoding is satisfiable, \ie, if the solver finds a model.  Constructing the
target program then amounts to obtaining the values for the templates from the
model. This approach is shown in Algorithm~\ref{alg:so}(\subref{alg:bso}).

\begin{figure}[t]
  \begin{subalgorithm}[b]{.5\textwidth}
    \begin{algorithmic}[1]
      \Function{BasicSo}{$p_s, C$}
      \State $n \gets 0$
      \While{true}
      \ForAll{$p_t \in \{ p \mid C(p) = n \}$}
      \State $\chi \gets \Call{EncodeBso}{p_s, p_t}$
      \If{\Call{Satisfiable}{$\chi$}}
      \State $m \gets \Call{GetModel}{\chi}$
      \State $p_t \gets \Call{DecodeBso}{m}$
      \State \textbf{return} $p_t$
      \EndIf
      \EndFor
      \State $n \gets n + 1$
      \EndWhile
      \EndFunction
    \end{algorithmic}
    \subcaption{Basic Superoptimization.}
    \label{alg:bso}
  \end{subalgorithm}
  \begin{subalgorithm}[b]{.5\linewidth}
    \begin{algorithmic}[1]
      \Function{UnboundedSo}{$p_s, C$}
      \State $p_t \gets p_s$
      \State $\chi \gets \Call{EncodeUso}{p_t} \land \Call{Bound}{p_t, C}$
      \While{\Call{Satisfiable}{$\chi$}}
      \State $m \gets \Call{GetModel}{\chi}$
      \State $p_t \gets \Call{DecodeUso}{m}$
      \State $\chi \gets \chi \land \Call{Bound}{p_t, C}$
      \EndWhile
      \State \textbf{return} $p_t$
      \EndFunction
    \end{algorithmic}
    \caption{Unbounded Superoptimization.}
    \label{alg:uso}
  \end{subalgorithm}
  \captionsetup{labelformat=alglabel}
  \caption{Superoptimization.}
  \label{alg:so}
\end{figure}

\emph{Unbounded superoptimization}~\cite{2002_joshi_et_al, 2017_jangda_et_al}
pushes this idea further.
Instead of searching through candidate programs and calling the \smt solver on
them, it shifts the search into the solver, \ie, the encoding expresses all
candidate instruction sequences of any length that correctly implement the
source program.
This approach is shown in Algorithm~\ref{alg:so}(\subref{alg:uso}): if the
solver returns satisfiable then there is an instruction sequence that correctly
implements the source program. Again, this target program is reconstructed from
the model. If successful, a constraint asking for a cheaper program is added and
the solver is called again.  Note that this also means that unbounded
superoptimization can stop with a correct, but possibly non-optimal solution. In
contrast, basic superoptimization cannot return a correct solution until
it has finished.

The main ingredients of superoptimization in Algorithm~\ref{alg:so}
are \textsc{Encode\-Bso/Uso} producing the \smt encoding, and
\textsc{DecodeBso/Uso} reconstructing the target program from a model.
We present our encodings for the semantics of \evm bytecode in the
following section.


\section{Encoding} \label{sec:encoding}

We start by encoding three parts of the \evm execution state:
\begin{enumerate*}[label=\emph{(\roman*)}]
\item the stack,
\item gas consumption, and
\item whether the execution is in an exceptional halting state.
\end{enumerate*}
We model the stack as an uninterpreted function together with a
counter, which points to the next free position on the stack.

\begin{definition}
  A state $\sigma = \langle \stack, \stackctr, \halt, \gas \rangle$
  consists of
  \begin{enumerate}[label=\emph{(\roman*)}]
  \item a function $\stack(\Vars,j,n)$ that, after the program has
    executed $j$ instructions on input variables from $\Vars$ returns
    the word from position $n$ in the stack,
  \item a function $\stackctr(j)$ that returns the number of words on the stack
    after executing $j$ instructions. Hence
    $\stack(\Vars,j,\stackctr(j) - 1)$ returns the top of the stack.
  \item a function $\halt(j)$ that returns true ($\top$) if
    exceptional halting has occurred after executing $j$ instructions,
    and false ($\bot$) otherwise.
  \item a function $\gas(\Vars, j)$ that returns the amount of gas
    consumed after executing $j$ instructions.
  \end{enumerate}
\end{definition}
Here the functions in $\sigma$ represent \emph{all} execution states of a
program, indexed by variable $j$.
\begin{example} \label{ex:exec}
  Symbolically executing the program \lstinline!PUSH 41 PUSH 1 ADD!
  using our representation above we have
  \begin{xalignat*}{4}
    \gas(0) &= 0  & \gas(1) &= 3 & \gas(2) &= 6 & \gas(3) &= 9\\
    \stackctr(0) &= 0 & \stackctr(1) &= 1 & \stackctr(2) &= 2 & \stackctr(3) &= 1\\
    \stack(1, 0) &= 41  & \stack(2, 0) &= 41 & \stack(2, 1) &= 1 & \stack(3, 0) &= 42
  \end{xalignat*}
  and $\halt(0) = \halt(1) = \halt(2) = \halt(3) = \bot$.
\end{example}

Note that this program does not consume any words that were already on
the stack. This is not the case in general. For instance we might be
dealing with the body of a function, which takes its arguments from
the stack.  Hence we need to ensure that at the beginning of the
execution sufficiently many words are on the stack. To this end we
first compute the \emph{depth} $\hat\delta(p)$ of the program~$p$,
\ie, the number of words a program $p$ consumes.  Then we take
variables $x_0, \ldots, x_{\hat\delta(p)-1}$ that represent the input
to the program and initialize our functions accordingly.
\begin{definition}
  \label{def:init}
  For a program with $\hat\delta(p) = d$ we initialize the state~$\sigma$ using
  \[ \gas_\sigma(0) = 0 \land \halt_\sigma(0) = \bot \land {\stackctr_\sigma(0) = d} \land
    {\bigwedge_{0 \leqslant \ell < d} \stack_\sigma(\Vars, 0, \ell) = x_\ell} \]
\end{definition}

For instance, for the program consisting of the single instruction \ADD\ we set
$\stackctr(0) = 2$, and $\stack(\{ x_0,x_1 \}, 0, 0) = x_0$ and
$\stack(\{ x_0,x_1 \}, 0, 1) = x_1$. We then have
$\stack(\{ x_0,x_1 \}, 1, 0) = x_1 + x_2$.

To encode the effect of \evm instructions we build \smt formulas to capture
their operational semantics. That is, for an instruction~$\iota$ and a state
$\sigma$ we give a formula~$\tau(\iota,\sigma, j)$ that defines the effect on
state $\sigma$ if $\iota$ is the $j$-th instruction that is executed. Since
large parts of these formulas are similar for every instruction and only depend
on $\delta$ and $\alpha$ we build them from smaller building blocks.

\begin{definition}
  For an instruction $\iota$ and state $\sigma$ we define:
  \begin{align*}
    \tau_\gas(\iota, \sigma, j) &\equiv \gas_\sigma(\Vars, j + 1) = \gas_\sigma(\Vars, j) + C(\sigma, j, \iota)\\
    \tau_\stackctr(\iota,\sigma,j) &\equiv \stackctr_{\sigma}(j + 1) = \stackctr_{\sigma}(j) + \alpha(\iota) - \delta(\iota)\\
    \tau_\pres(\iota,\sigma,j) &\equiv \forall\, n. n < \stackctr_{\sigma}(j) - \delta(\iota)
                                 \rightarrow  \stack_{\sigma}(\Vars,j + 1,n) = \stack_{\sigma}(\Vars,j,n) \\
    \tau_\halt(\iota,\sigma,j) &\equiv \halt_{\sigma}(j + 1) = \halt_{\sigma}(j) \lor \stackctr_{\sigma}(j) - \delta(\iota) < 0
                                 \lor \stackctr_{\sigma}(j) - \delta(\iota) + \alpha(\iota) > 2^{10}
  \end{align*}
  Here $C(\sigma, j, \iota)$ is the gas cost of executing instruction
  $\iota$ on state $\sigma$ after $j$ steps.
\end{definition}
The formula~$\tau_\gas$ adds the cost of $\iota$ to the gas cost incurred so
far. The formula~$\tau_\stackctr$ updates the counter for the number of words on
the stack according to $\delta$ and $\alpha$. The formula~$\tau_\pres$ expresses
that all words on the stack below $\stackctr_{\sigma}(j) - \delta(\iota)$ are
preserved. Finally, $\tau_\halt$ captures that exceptions relevant to the stack
can occur through either an underflow or an overflow, and that once it has occurred an
exceptional halt state persists.
For now the only other component we need is how the instructions affect the
stack $\stack$, \ie, a formula $\tau_\stack(\iota, \sigma, j)$.  Here we only
give an example and refer to our implementation or the yellow
paper~\cite{2018_wood} for details. We have
\begin{align*}
  \tau_\stack(\ADD, \sigma, j)
  &\equiv \stack_\sigma(\Vars, j+1, \stackctr_{\sigma}(j + 1) - 1)\\
  &= \stack_\sigma(\Vars, j, \stackctr_{\sigma}(j) - 1)
    + \stack_\sigma(\Vars, j, \stackctr_{\sigma}(j) - 2)
\end{align*}
Finally these formulas yield an encoding for the semantics of an instruction.
\begin{definition}
  For an instruction $\iota$ and state $\sigma$ we define
  \begin{align*}
    \tau(\iota, \sigma, j)  \equiv \tau_\stack(\iota, \sigma, j)
    \land \tau_\stackctr(\iota,\sigma, j)
    \land \tau_\gas(\iota,\sigma, j)
    \land \tau_\halt(\iota,\sigma,j)
    \land \tau_\pres(\iota,\sigma,j)
  \end{align*}
\end{definition}

Then to encode the semantics of a program~$p$ all we need to do is
to apply $\tau$ to the instructions of $p$.

\begin{definition}
  For a program $p = \iota_0 \cdots \iota_n$ we set
  $\tau(p,\sigma) \equiv \bigwedge_{0\leqslant j \leqslant n}
  \tau(\iota_j,\sigma,j)$.
\end{definition}

Before building an encoding for superoptimization we consider another
aspect of the \evm for our state representation: storage and
memory. The gas cost for storing words depends on the words that are
currently stored. Similarly, the cost for using memory depends on the
number of bytes currently used. This is why the cost of an instruction
$C(\sigma, j, \iota)$ depends on the state and the function
$\gas_\sigma$ accumulating gas cost depends on $\Vars$.

To add support for storage and memory to our encoding there are two natural
choices: the theory of arrays or an Ackermann encoding. However, since we have
not used arrays so far, they would require the solver to deal with an additional
theory. For an Ackermann encoding we only need uninterpreted functions, which we
have used already.
Hence, to represent storage in our encoding we extend states with an
uninterpreted function $\storage(\Vars, j, k)$, which returns the word at
key~$k$ after the program has executed $j$~instructions.
Similarly to how we set up the initial stack we need to deal with the values
held by the storage before the program is executed.
Thus, to initialize $\storage$ we introduce fresh variables to
represent the initial contents of the storage. More precisely, for all
\SLOAD and \SSTORE instructions occurring at positions
$j_1, \ldots, j_\ell$ in the source program, we introduce fresh
variables $s_1, \ldots, s_\ell$ and add them to $\Vars$. Then for a
state $\sigma$ we initialize $\storage_\sigma$ by adding the following
conjunct to the initialization constraint from
Definition~\ref{def:init}:
\[
  \forall w.\ \storage_\sigma(\Vars, 0, w) = \ite(
  w = a_{j_1}, s_1, \ite(w = a_{j_2}, s_2,\ldots, \ite(w = a_{j_\ell}, s_\ell, w_{\bot})))
\]
where $a_j = \stack_\sigma(\Vars, j, \stackctr(j) - 1)$ and $w_\bot$ is
the default value for words in the storage.

The effect of the two storage instructions $\SLOAD$ and $\SSTORE$ can then be
encoded as follows:
\begin{align*}
  \tau_\stack(\SLOAD, \sigma, j)
  &\equiv \stack_\sigma(\Vars, j+1, \stackctr_\sigma(j + 1) - 1) = \storage(\Vars, j, \stack_\sigma(\Vars, j, \stackctr_\sigma(j) - 1))\\
    \tau_\storage(\SSTORE, \sigma, j)
  &\equiv \forall w.\ \storage_\sigma(\Vars, j+1, w) =\\
  &\quad\ite(w = \stack_\sigma(\Vars, j, \stackctr_\sigma(j) - 1),  \stack_\sigma(\Vars, j, \stackctr_\sigma(j) - 2), \storage_\sigma(\Vars, j, w))
\end{align*}
Moreover all instructions except $\SSTORE$ preserve the storage,
that is, for $\iota \neq \SSTORE$ we add the following conjunct to
$\tau_\pres(\iota, \sigma, j)$:
\[ \forall w.\ \storage_\sigma(\Vars, j + 1, w) = \storage_\sigma(\Vars, j, w) \]

To encode memory a similar strategy is an obvious way to go.  However,
we first want to evaluate the solver's performance on the encodings
obtained when using stack and storage.  Since the solver already
struggled, due to the size of the programs and the number of
universally quantified variables, see Section~\ref{sec:evaluation}, we
have not yet added an encoding of memory.

Finally, to use our encoding for superoptimization we need an encoding of
equality for two states after a certain number of instructions. Either to ensure
that two programs are equivalent (they start and end in equal states) or
different (they start in equal states, but end in different ones). The following
formula captures this constraint.
\begin{definition}
  For states $\sigma_1$ and $\sigma_2$ and program locations $j_1$ and $j_2$ we define
  \begin{align*}
    \epsilon(\sigma_1,\sigma_2, j_1, j_2)
    &\equiv
      \stackctr_{\sigma_1}(j_1) = \stackctr_{\sigma_2}(j_2)
      \land \halt_{\sigma_1}(j_1) = \halt_{\sigma_2}(j_2)\\
    &\land \forall\, n. n <
      \stackctr_{\sigma_1}(j_1) \rightarrow \stack_{\sigma_1}(\Vars,j_1,n) =
      \stack_{\sigma_2}(\Vars,j_2,n)\\
    &\land \forall\, w. \storage_{\sigma_1}(\Vars, j_1, w) = \storage_{\sigma_2}(\Vars, j_2, w)
  \end{align*}
\end{definition}
Since we aim to improve gas consumption, we do not demand equality for $\gas$.

We now have all ingredients needed to implement basic superoptimization:
simply enumerate all possible programs ordered by gas cost and use the encodings
to check equivalence. However, since already for one \lstinline!PUSH! there are
$2^{256}$ possible arguments, this will not produce results in a reasonable
amount of time. Hence we use templates as described in
Section~\ref{sec:superopt}. We introduce an uninterpreted function $\tmpl(j)$
that maps a program location $j$ to a word, which will be the argument of
\lstinline!PUSH!. The solver then fills these templates and we can
get the values from the model. This is a step forward, but since we have
\num{80} encoded instructions, enumerating all permutations still yields too
large a search space.
Hence we use an encoding similar to the CEGIS
algorithm~\cite{2011_gulwani_et_al}. Given a collection of
instructions, we formulate a constraint representing all possible
permutations of these instructions. It is satisfiable if there is a
way to connect the instructions into a target program that is
equivalent to the source program. The order of the instructions can
again be reconstructed from the model provided by the solver.
More precisely given a source program $p$ and a list of candidate instructions
$\iota_1,\ldots,\iota_n$, \textsc{EncodeBso} from Algorithm~\ref{alg:so}(\subref{alg:bso})
takes variables $j_1,\ldots,j_n$ and two states $\sigma$ and
$\sigma'$ and builds the following formula
\begin{align*}
  \forall \Vars.\,
  &\epsilon(\sigma,\sigma',0,0) \land \epsilon(\sigma,\sigma',|p|,n)
    \land \tau(p,\sigma)\\
  & \land \bigwedge_{1 \leqslant \ell \leqslant n} \tau(\iota_\ell,\sigma',j_\ell)
    \land \bigwedge_{1\leqslant \ell < k \leqslant n} j_\ell \neq j_k
    \land \bigwedge_{1\leqslant \ell \leqslant n} j_\ell \geqslant 0 \land j_\ell < n
\end{align*}
Here the first line encodes the source program, and says that the
start and final states of the two programs are equivalent. The second
line encodes the effect of the candidate instructions and enforces
that they are all used in some order. If this formula is satisfiable
we can simply get the $j_i$ from the model and reorder the candidate
instructions accordingly to obtain the target program.

Unbounded superoptimization shifts even more of the search into the
solver, encoding the search space of all possible programs.
To this end we take a variable~$n$, which represents the number of
instructions in the target program and an uninterpreted function
$\instr(j)$, which acts as a template, returning the instruction to be
used at location $j$. Then, given a set of candidate instructions the
formula to encode the search can be built as follows:
\begin{definition}
  Given a set of instructions $\CI$ we define the formula $\rho(\sigma,n)$
  as
  \begin{align*}
  \forall j.\, j \geqslant 0 \land j < n
  &\rightarrow \bigwedge_{\iota \in \CI} \instr(j) = \iota \to \tau(\iota,\sigma,j)
    \land \bigvee_{\iota \in \CI} \instr(j) = \iota
\end{align*}
Finally, the constraint produced by \textsc{EncodeUso} from
Algorithm~\ref{alg:so}(\subref{alg:uso}) is
\begin{align*}
  \forall \Vars.\,
  &\tau(p,\sigma) \land \rho(\sigma', n) \land \epsilon(\sigma,\sigma', 0, 0)
  \land \epsilon(\sigma,\sigma',|p|,n) \land \gas_\sigma(\Vars, |p|) > \gas_{\sigma'}(\Vars, n)
\end{align*}
\end{definition}

During our experiments we observed that the solver struggles to show that the
formula is unsatisfiable when $p$ is already optimal. To help in these cases we
additionally add a bound on $n$: since the cheapest \evm instruction has gas
cost~\num{1}, the target program cannot use more instructions than the gas
cost of $p$, \ie, we add $n \leqslant \gas_\sigma(\Vars, |p|)$.

In our application domain there are many instructions that fetch
information from the outside world. For instance, \ADDRESS\ gets the
\ethereum address of the account currently executing the bytecode of
this smart contract. Since it is not possible to know these values at
compile time we cannot encode their full semantics. However, we would
still like to take advantage of structural optimizations where these
instructions are involved, \eg, via \lstinline!DUP! and \lstinline!SWAP!.

\begin{example}
  \label{ex:address-dup}
  Consider the program \lstinline!ADDRESS DUP1!. The same effect can be achieved
  by simply calling \lstinline!ADDRESS ADDRESS!. Duplicating words on the stack,
  if they are used multiple times, is an intuitive approach. However, because
  executing $\ADDRESS$ costs \SI{2}{\gas} and $\DUP{1}$ costs \SI{3}{\gas},
  perhaps unexpectedly, the second program is cheaper.
\end{example}

To find such optimizations we need a way to encode $\ADDRESS$ and
similar instructions. For our purposes, these instructions have in
common that they put arbitrary but fixed words onto the
stack. Analogous to uninterpreted functions, we call them
\emph{uninterpreted instructions} and collect them in the set $\UI$.
To represent their output we use universally quantified
variables---similar to input variables.  To encode the effect
uninterpreted instructions have on the stack, \ie, $\tau_\stack$, we
distinguish between \emph{constant} and \emph{non-constant}
uninterpreted instructions.

Let $\ui_c(p)$ be the set of \emph{constant uninterpreted
instructions} in $p$, \ie
$\ui_c(p) = \{ \iota \in p \mid \iota \in \UI \land \delta(\iota) = 0
\}$.  Then for $\ui_c(p) = \{\iota_1, \ldots, \iota_k\}$ we take
variables $u_{\iota_1},\ldots,u_{\iota_k}$ and add them to $\Vars$,
and thus to the arguments of the state function~$\stack$. The formula
$\tau_\stack$ can then use these variables to represent the unknown
word produced by the uninterpreted instruction, \ie, for
$\iota \in \ui_c(p)$ with the corresponding variable $u_\iota$ in
$\Vars$, we set
$\tau_\stack(\iota, \sigma, j) \equiv \stack_\sigma(\Vars, j + 1,
\stackctr_\sigma(j)) = u_{\iota}$.

For a \emph{non-constant instruction}~$\iota$, such as \BLOCKHASH\ or
\BALANCE, the word put onto the stack by $\iota$ depends on the top
$\delta(\iota)$ words of the stack. We again model this dependency
using an uninterpreted function. That is, for every non-constant
uninterpreted instruction~$\iota$ in the source program~$p$,
$\ui_{n}(p) = \{ \iota \in p \mid \iota \in \UI \land \delta(\iota) >
0 \}$, we use an uninterpreted function~$f_\iota$. Conceptually, we
can think of $f_\iota$ as a read-only memory initialized with the
values that the calls to $\iota$ produce.

\begin{example}
  \label{ex:blockhash}
  The instruction $\BLOCKHASH$ gets the hash of a given block~$b$.  Thus
  optimizing the program %
  \lstinline[mathescape=true]!PUSH $b_1$ BLOCKHASH PUSH $b_2$ BLOCKHASH!
  depends on the values $b_1$
  and $b_2$. If $b_1 = b_2$ then the cheaper program
  \lstinline[mathescape=true]!PUSH $b_1$ BLOCKHASH DUP1!
  yields the same state as the original program.
\end{example}

To capture this behaviour, we need to associate the arguments $b_1$ and $b_2$ of
\BLOCKHASH\ with the two different results they may produce. As with constant
uninterpreted instructions, to model arbitrary but fixed results, we add fresh
variables to $\Vars$.  However, to account for different results produced by
$\ell$ invocations of $\iota$ in $p$ we have to add $\ell$ variables.
Let $p$ be a program and $\iota \in \ui_n(p)$ a unary instruction which appears
$\ell$ times at positions $j_1, \ldots, j_\ell$ in $p$. For variables
$u_{1}, \ldots, u_{\ell}$, we initialize $f_\iota$ as follows:
\begin{align*}
  \forall w.\, f_\iota(\Vars, w) = \ite(w = a_{j_1}, u_1,
  \ite(w = a_{j_2}, u_2, \ldots,
  \ite(w = a_{j_\ell}, u_\ell, w_{\bot})))
\end{align*}
where $a_j$ is the word on the stack after $j$ instructions in~$p$,
that is $a_j = \stack_\sigma(\Vars, j, \stackctr(j)-1)$,
and $w_\bot$ is a default word.

This approach straightforwardly extends to instructions with more than one
argument. Here we assume that uninterpreted instructions put exactly one word
onto the stack, \ie, $\alpha(\iota) = 1$ for all $\iota \in \UI$. This
assumption is easily verified for the \evm: the only instructions with
$\alpha(\iota) > 1$ are \DUP{} and \SWAP{}.  Finally we set the effect a
non-constant uninterpreted instruction $\iota$ with associated function
$f_\iota$ has on the stack:
\[
  \tau_\stack(\iota, \sigma, j) \equiv
  \stack_\sigma(\Vars, j+1, \stackctr_\sigma(j + 1) - 1) =
  f_{\iota}(\Vars,  \stack_\sigma(\Vars, j, \stackctr_\sigma(j) - 1))
\]

For some uninterpreted instructions there might a be way to partially encode
their semantics. The instruction \BLOCKHASH returns \pusharg{0} if it is called
for a block number greater than the current block number. While the current
block number is not known at compile time, the instruction \NUMBER does return
it. Encoding this interplay between \BLOCKHASH and \NUMBER could potentially be
exploited for finding optimizations.


\section{Implementation} \label{sec:implementation}

We implemented basic and unbounded superoptimization in our
tool~\ebso, which is available under the Apache-2.0 license:
\href{https://github.com/juliannagele/ebso}{\texttt{github.com/juliannagele/ebso}}.
The encoding employed by \ebso uses several background theories:
\begin{enumerate*}[label=\emph{(\roman*})]
\item uninterpreted functions (UF) for encoding the state of the \evm,
  for templates, and for encoding uninterpreted instructions,
\item bit vector arithmetic (BV) for operations on words,
\item quantifiers for initial words on the stack and in the storage, and the
  results of uninterpreted instructions, and
\item linear integer arithmetic (LIA) for the instruction counter.
\end{enumerate*}
Hence following the SMT-LIB classification%
\footnote{\href{http://smtlib.cs.uiowa.edu/logics.shtml}{smtlib.cs.uiowa.edu/logics.shtml}}
\ebso's constraints fall under the logic UFBVLIA. As \smt solver we
chose \ZE~\cite{2008_de_moura_et_al}, version~4.7.1 which we call with
default configurations. In particular, \ZE performed well for the
theory of quantified bit vectors and uninterpreted functions in the
last \smt competition (albeit non-competing).\footnote{%
  \href{https://smt-comp.github.io/2019/results/ufbv-single-query}
  {smt-comp.github.io/2019/results/ufbv-single-query}}

The aim of our implementation is to provide a prototype
without relying on heavy engineering and optimizations
such as exploiting parallelism or tweaking \ZE strategies.
But without any optimization, for the full word size of the
\evm---\SI{256}{\bit}---\ebso did not handle the simple program
\lstinline!PUSH 0 ADD POP! within a reasonable amount of time. Thus
we need techniques to make \ebso viable.
By investigating the models generated by \ZE run with the default
configuration, we believe that the problem lies with the leading
universally quantified variables. And we have plenty of them: for the
input on the stack, for the storage, and for uninterpreted
instructions.
By reducing the word size to a small $k$, we can reduce the search space for
universally quantified variables from $2^{256}$ to some significantly
smaller $2^k$. But then we need to check any target program found
with a smaller word size.

\begin{example}
  The program
  \lstinline!PUSH 0 SUB PUSH 3 ADD! %
  from Example~\ref{ex:intro}
  optimizes to \lstinline!NOT! for word size \SI{2}{\bit}, because
  then the binary representation of \lstinline!3! is all ones. When
  using word size \SI{256}{\bit} this optimization is not correct.
\end{example}

To ensure that the target program has the same semantics for word
size~\SI{256}{\bit}, we use \emph{translation validation}: we ask the
solver to find inputs, which distinguish the source and target
programs, \ie, where both programs start in equivalent states, but
their final state is different. Using our existing machinery this
formula is easy to build:%
\footnote{This approach also allows for other over-approximations. For instance,
  we tried using integers instead of bit vectors, which performed worse.}
\begin{definition} %
  Two programs $p$ and $p'$ are equivalent if
  \[ \TV(p, p', \sigma, \sigma') \equiv %
    {\exists \Vars, \tau(p,\sigma) \land \tau(p',\sigma') \land
      \epsilon(\sigma,\sigma',0,0) \land \lnot
      \epsilon(\sigma,\sigma',|p|,|p'|)} \] %
  is unsatisfiable. Otherwise, $p$ and $p'$ are different, and the
  values for the variables in $\Vars$ from the model are a
  corresponding witness.
\end{definition}

A subtle problem remains: how can we represent the program
\lstinline!PUSH 224981! with only $k$~bit? Our solution is to replace
arguments $a_1, \ldots, a_m$ of \lstinline!PUSH! where
$a_i \geqslant 2^k$ with fresh, universally quantified variables
$c_1, \ldots, c_m$.
If a target program is found, we replace $c_i$ by the original value
$a_i$, and check with translation validation whether this target
program is correct.
A drawback of this approach is that we might lose potential
optimizations.
\begin{example}
  The program \lstinline!PUSH 0b111...111 AND! optimizes to the empty
  program. But, abstracting the argument of \lstinline!PUSH!
  translates the program to \lstinline[mathescape=true]!PUSH $c_i$ AND!, which does not
  allow the same optimization.
\end{example}

Like many compiler optimizations, \ebso optimizes basic blocks.  Therefore we
split \evm bytecode along instructions that change the control flow, \eg
\lstinline!JUMPI!, or \lstinline!SELFDESTRUCT!.  Similarly we further split
basic blocks into (\ebso) blocks so that they contain only encoded instructions.
Instructions, which are not encoded, or encodable, include
instructions that write to memory, \eg \lstinline!MSTORE!, or
the log instructions \lstinline!LOG!.

\begin{lemma} \label{lem:closed-under-context} %
  If program $p$ superoptimizes to program $t$ then in any program we
  can replace $p$ by $t$.
\end{lemma}

\begin{proof} %
  We show the statement by induction on the program context
  $(c_1, c_2)$ of the program $c_1 p c_2$. By assumption, the
  statement holds for the base case $([\ ], [\ ])$. For the step case
  $(\iota c_1, c_2)$, we observe that every instruction~$\iota$ is
  deterministic, \ie executing $\iota$ starting from a state~$\sigma$
  leads to a deterministic state~$\sigma'$. By induction hypothesis,
  executing ${c_1 p c_2}$ and ${c_1 t c_2}$ from a state~$\sigma'$
  leads to the same state $\sigma''$, and therefore we can replace
  ${{\iota c_1} p c_2}$ by ${{\iota c_1} t c_2}$. We can reason
  analogously for $(c_1, c_2 \iota)$.
\end{proof}


\section{Evaluation}\label{sec:evaluation}

We evaluated \ebso on two real-word data sets:
\begin{enumerate*}[label=\emph{(\roman*})]
\item optimizing an already highly optimized data set in
  Section~\ref{sec:gg:evaluation}, and
\item a large-scale data set from the \ethereum blockchain to compare
  basic and unboundend superoptimization in
  Section~\ref{sec:bc:evaluation}.
\end{enumerate*}
We use \ebso to extract \ebso blocks from our data sets. From the
extracted blocks
\begin{enumerate*}[label=\emph{(\roman*)}]
\item we remove duplicate blocks, and
\item we remove blocks which are only different in the arguments of
  \lstinline!PUSH! by abstracting to word size \SI{4}{bit}.
\end{enumerate*}
We run both evaluations on a cluster~\cite{2017_king_et_al} consisting
of nodes running Intel Xeon E5645 processors at
\SI{2.40}{\giga\hertz}, with one core and \SI{1}{\gibi\byte} of memory
per instance.

We successfully validated all optimizations found by \ebso by running
a reference implementation of the \evm on pseudo-random
input. Therefore, we run the bytecode of the original input block and
the optimized bytecode to observe that both produce the same final
state. The \evm implementation we use is
\bsh{go-ethereum}\footnote{\href{https://github.com/ethereum/go-ethereum}{github.com/ethereum/go-ethereum}}
version \bsh{1.8.23}.

\subsection{Optimize the Optimized} \label{sec:gg:evaluation}

This evaluation tests \ebso against human intelligence. Underlying our
data set are \num{200}~\solidity contracts ($\ggRaw$) we collected from
the \emph{1st Gas Golfing Contest}.%
\footnote{\href{https://g.solidity.cc/}{g.solidity.cc}}
In that contest competitors had to write the most gas-efficient \solidity
code for five given challenges: %
\begin{enumerate*}[label=\emph{(\roman*)}]
\item integer sorting,
\item implementing an interpreter,
\item hex decoding,
\item string searching, and
\item removing duplicate elements.
\end{enumerate*}
Every challenge had two categories: \emph{standard} and \emph{wild}. For
wild, any \solidity feature is allowed---even inlining \evm
bytecode. The winner of each track received \SI{1}{Ether}.
The Gas Golfing Contest provides a very high-quality data set: the
\evm bytecode was not only optimized by the \solc compiler, but also
by humans leveraging these compiler optimizations and writing inline
code themselves.
To collect our data set~$\ggData$, we first compiled the \solidity
contracts in $\ggRaw$ with the same set-up as in the contest.%
\footnote{Namely, \bsh{\$ solc --optimize --bin-runtime
    --optimize-runs 200} with \solc compiler version~\bsh{0.4.24}
  available at
  \href{https://github.com/ethereum/solidity/tree/v0.4.24}%
  {github.com/ethereum/solidity/tree/v0.4.24}.} One contract in the
wild category failed to compile and was thus excluded
from~$\ggRaw$.
From the generated \bsh{.bin-runtime} files, we extracted our
final data set~$\ggData$ of $\ggBlockcount$~distinct blocks.

For this evaluation, we run \ebso in its default mode: unbounded
superoptimization. We run unbounded superoptimization because, as can be seen in
Section~\ref{sec:bc:evaluation}, in our context unbounded superoptimization
outperformed basic superoptimization. As time-out for this challenging
data set, we estimated \SI{1}{\hour} as reasonable.

\begin{table}[t]
  \setlength{\tabcolsep}{6pt}
  \centering
  \begin{tabular}{ r r r}
\toprule
& \# & \SI{}{\percent} \\

\midrule
optimized (optimal) & \num{19} (\num{10}) & \SI{0.69}{\percent} (\SI{0.36}{\percent}) \\
proved optimal & \num{481} & \SI{17.54}{\percent} \\
time-out (trans.\ val.\ failed) & \num{2243} (\num{196}) & \SI{81.77}{\percent} (\SI{7.15}{\percent}) \\
\bottomrule\\
\end{tabular}

  \caption{Aggregated results of running \ebso on $\ggData$.}
  \label{tab:gg:results}
\end{table}
Table~\ref{tab:gg:results} shows the aggregated results of running
\ebso on $\ggData$. In total, \ebso optimizes \num{19}~blocks out of
\ggBlockcount, \num{10}~of which are shown to be optimal. %
%
Moreover, \ebso can prove for more than \SI{17}{\percent} of blocks in
$\ggData$ that they are already optimal.
It is encouraging that \ebso even finds optimizations in this already highly
optimized data set. The quality of the data set is supported by the high
percentage of blocks being proved as optimal by \ebso.
Next we examine three found optimizations more closely.
Our favorite optimization
\lstinline!POP PUSH 1 SWAP1 POP PUSH 0! to %
\lstinline!SLT DUP1 EQ PUSH 0! %
witnesses that superoptimization can find unexpected results, and that
unbounded superoptimization can stop with non-optimal results:
\lstinline!SLT DUP1 EQ! is, in fact, a round-about and optimizable way
to pop two words from the stack and push \pusharg{1} on the
stack.  Some optimizations follow clear patterns.
The optimizations  %
\lstinline!CALLVALUE DUP1 ISZERO PUSH 81! to %
\lstinline!CALLVALUE CALLVALUE ISZERO PUSH 81! %
and \lstinline!CALLVALUE DUP1 ISZERO PUSH 364! to %
\lstinline!CALLVALUE CALLVALUE ISZERO PUSH 364! %
are both based on the fact that \lstinline!CALLVALUE! is cheaper than
\DUP{1}. Finding such patterns and generalizing them into peephole
optimization rules could be interesting future work.

Unfortunately, \ebso hit a time-out in nearly \SI{82}{\percent} of all cases,
where we count a failed translation validation as part of the time-outs, since
in that case \ebso continues to search for optimizations after increasing the
word size.



\subsection{Unbounded \vs Basic Superoptimization}
\label{sec:bc:evaluation}

In this evaluation we compare unbounded and basic superoptimization,
which we will abbreviate with \uso and \bso, respectively. To compare
\uso and \bso, we want a considerably larger data set. Fortunately,
there is a rich source of \evm bytecode accessible: contracts deployed
on the \ethereum blockchain. Assuming that contracts that are called
more often are well constructed, we queried the \num{2500} most called
contracts\footnote{up to block number \num{7300000} deployed on
  Mar-04-2019 01:22:15 AM +UTC}
using \google
BigQuery.\footnote{\href{https://cloud.google.com/blog/products/data-analytics/ethereum-bigquery-public-dataset-smart-contract-analytics}
  {cloud.google.com/blog/products/data-analytics/ethereum-bigquery-public-dataset-smart-contract-analytics}}
From them we extract our data set~$\bcData$ of $\bcBlockcount$
distinct blocks.
For this considerably larger data set, we estimated a cut-off point of
\SI{15}{\min} as reasonable. One limitation is that, due to the high volume, we
only run the full evaluation once.

\begin{table}[t]
  \setlength{\tabcolsep}{6pt}
  \centering
  \begin{tabular}{ r r r r r}
\toprule
& \multicolumn{2}{l}{$\uso$} & \multicolumn{2}{l}{$\bso$} \\
& \# & \SI{}{\percent} & \# & \SI{}{\percent} \\

\midrule
optimized (optimal) & \num{943} (\num{393}) & \SI{1.54}{\percent} (\SI{0.64}{\percent}) & \num{184} & \SI{0.3}{\percent} \\
proved optimal & \num{3882} & \SI{6.34}{\percent} & \num{348} & \SI{0.57}{\percent} \\
time-out (trans.\ val.\ failed) & \num{56392} (\num{1467}) & \SI{92.12}{\percent} (\SI{2.4}{\percent}) & \num{60685} & \SI{99.13}{\percent} \\
\bottomrule\\
\end{tabular}

  \caption{Aggregated results of running \ebso with $\uso$ and $\bso$
    on $\bcData$.}
  \label{tab:bc:stats}
\end{table}

Table~\ref{tab:bc:stats} shows the aggregated results of running \ebso on
$\bcData$. Out of $\bcBlockcount$ blocks in $\bcData$, \ebso
finds~\num{943}~optimizations using \uso out of which it proves~\num{393} to be
optimal. Using~\bso~\num{184} optimizations are found. Some blocks were shown to
be optimal by both approaches. Also, both approaches time out in a majority of
the cases: \uso in more than \SI{92}{\percent}, and \bso in more than
\SI{99}{\percent}.
%
Over all $\bcBlockcount$ blocks the total amount of gas saved for \uso is
$\bcUsoTotal$ and $\bcBsoTotal$ for \bso. For all blocks where an optimization
is found, the average gas saving per block in \uso is $\bcUsoRelative$, and
$\bcBsoRelative$ for \bso. The higher average for \bso can be explained by
\begin{enumerate*}[label=\emph{(\roman*})]
\item \bso's bias for smaller blocks, where relative savings are
  naturally higher, and
\item \bso only providing optimal results, whereas \uso may find
  intermediate, non-optimal results.
\end{enumerate*}
The optimization with the largest gain, is one which we did not
necessarily expect to find in a deployed contract: a redundant storage
access. Storage is expensive, hence optimized for in deployed
contracts, but \uso and \bso both found~\bcUsoMaxSource which
optimizes to the empty program---because the program basically loads
the value from key~4 only to store it back to that same
key. This optimization saves at least \bcUsoMaxGas, but up to
\SI{20220}{\gas}.

From Table~\ref{tab:bc:stats} we see that on $\bcData$, \uso
outperforms \bso by roughly a factor of five on found optimizations;
more than ten times as many blocks are proved optimal by \uso than by
\bso. %
As we expected, most optimizations found by \bso were also found by
\uso, but surprisingly, \bso found~\bcBsoOnly~optimizations, on which
\uso failed. We found that nearly all of the~\bcBsoOnly~source programs are
fairly complicated, but have a short optimization of two or three
instructions. To pick an example, the block %
\lstinline!PUSH 0 PUSH 12 SLOAD LT ISZERO ISZERO ISZERO PUSH 12250! %
is optimized to the relatively simple %
\lstinline!PUSH 1 PUSH 12250!---a candidate block, which will be tried
early on in \bso.  Additionally, all~\num{21} blocks are cheap: all
cost less than \SI{10}{\gas}.  We also would have expected at least
some of these optimizations to have been found by \uso. We believe
internal unfortunate, non-deterministic choice within the solver to be
the reason that it did not.



\section{Conclusion} \label{sec:conclusion}

\paragraph{Summary.} %
We develop \ebso, a superoptimizer for \evm bytecode, implementing two
different superoptimization approaches and compare them on a large set
of real-world smart contracts. Our experiments show that, relying on
the heavily optimized search heuristics of a modern \smt solver is a
feasible approach to superoptimizing \evm bytecode.
\paragraph{Related Work.} %
Superoptimization~\cite{1987_massalin} has been explored for a variety
of different contexts~\cite{2002_joshi_et_al, 2013_schkufza_et_al,
  2016_phothilimthana_et_al, 2017_jangda_et_al}, including binary
translation~\cite{2008_bansal_et_al} and synthesizing compiler
optimizations~\cite{2017_sasnauskas_et_al}. To our knowledge \ebso is
the first application of superoptimization to smart contracts.

Chen et al.~\cite{2018_chen_et_al} also aim to save gas by optimizing
\evm bytecode. They identified 24~anti patterns by manual
inspection. Building on their work we run \ebso on their identified
patterns. For 19~instances, \ebso too found the same
optimizations. For 2 patterns, \ebso lacks encoding of the
instructions (\STOP, \JUMP), and for 2~patterns \ebso times out on a
local machine.

Due to the repeated exploitation of flaws in smart contracts, various formal
approaches for analyzing \evm bytecode have been proposed.  For instance
Oyente~\cite{2016_luu_et_al} performs control flow analysis in order to detect
security defects such as reentrancy bugs.
\paragraph{Outlook.} %
There is ample opportunity for future work. We do not yet support the
\evm's memory. While conceptually this would be a straightforward
extension, the number of universally quantified variables and size of
blocks are already posing challenges for performance, as we identified
by analyzing the optimizations found by \ebso.

Thus, it would be interesting to use \smt benchmarks obtained by
\ebso's superoptimization encoding to evaluate different solvers, \eg
\CVCF\footnote{\href{http://cvc4.cs.stanford.edu/web}{cvc4.cs.stanford.edu/web/}}
or
\vampire\footnote{\href{http://www.vprover.org/}{www.vprover.org}}. %
The basis for this is already in place: \ebso can export the generated
constraints in SMT-LIB format. Accordingly, we plan to generate new
SMT benchmarks and submit them to one of the suitable categories of
SMT-LIB.

In order to ease the burden on developers \ebso could benefit from
caching common optimization patterns~\cite{2017_sasnauskas_et_al} to
speed up optimization times. Another fruitful approach
could be to extract the optimization patterns and generalize them into
peephole optimizations and rewrite rules.
%


\bibliographystyle{splncs04}
\bibliography{refs}
\newpage
\appendix

\end{document}